\documentclass[a4paper,12pt,oneside,reqno]{amsart}

\usepackage{hyperref}
\usepackage[headinclude,DIV13]{typearea}
\areaset{15.1cm}{25.0cm}
\usepackage{amsfonts,amssymb,amsmath,amsthm,bbm,longtable,verbatim}
\usepackage[latin1] {inputenc}
\usepackage{graphicx, psfrag}
\usepackage{color}
\usepackage{bold-extra}

\DeclareMathOperator{\sign}{sgn}

\newtheorem{theorem}{Theorem}[section]
\newtheorem{lemma}[theorem]{Lemma}
\newtheorem{proposition}[theorem]{Proposition}

\theoremstyle{definition}
\newtheorem{definition}[theorem]{Definition}

\theoremstyle{remark}
\newtheorem*{remark}{Remark}

\def\paragraph#1{\noindent \textbf{#1}}

\numberwithin{equation}{section}

\newcommand{\D}{{\mathbb D}}
\newcommand{\C}{{\mathbb C}}
\newcommand{\E}{{\mathbb E}}
\renewcommand{\d}{\partial{}}
\newcommand{\db}{\bar{\partial}}
\newcommand{\zb}{\bar{z}}
\newcommand{\tb}{{\beta_1}}
\newcommand{\tg}{{\gamma_1}}
\newcommand{\tp}{{\psi_1}}
\newcommand{\F}{{_2F_1}}

\renewcommand{\Re}{\mathrm {Re\,}}

\newcommand{\R}{{\mathbb R}}
\newcommand{\T}{{\mathbb T}}

\def\br#1{\left(#1\right)}
\def\brb#1{\left[#1\right]}

\newcommand{\K}{K}

\newcommand{\I}{\mathrm{I}}
\newcommand{\II}{\mathrm{II}}
\newcommand{\III}{\mathrm{III}}
\newcommand{\IV}{\mathrm{IV}}

\def\qqq#1{{\bf ????~{#1}~????}}

\begin{document}

\title{Integral means spectrum of whole-plane SLE}
\date{}
%\author{D. Beliaev, B. Duplantier, M. Zinsmeister}
 \author[D. Beliaev]{Dmitry Beliaev}
      \address{D. Beliaev\\Mathematical Institute\\
University of Oxford\\Andrew Wiles Building\\Radcliffe Observatory Quarter
\\Woodstock Road\\
Oxford
OX2 6GG\\ UK}
\email{belyaev@maths.ox.ac.uk}
   \author[B. Duplantier]{Bertrand Duplantier}
      \address{B. Duplantier\\Institut de Physique Th\'eorique\\ Universit\'e Paris-Saclay, CEA, CNRS\\ B\^at. 774, Orme des Merisiers\\ F-91191 Gif-sur- Yvette Cedex\\ France}
\email{bertrand.duplantier@cea.fr}
\author[M. Zinsmeister]{Michel Zinsmeister}
   \address{M. Zinsmeister\\MAPMO\\ Universit\'e d'Orl\'eans\\ B\^atiment de math\'ematiques\\ rue de Chartres B.P. 6759-F-45067 Orl\'eans Cedex 2, France}
     \email{zins@univ-orleans.fr}        

 \begin{abstract}
We complete the mathematical analysis of the fine structure of harmonic measure on SLE curves that was initiated in Ref. \cite{BeSmSLE}, as described by the averaged integral means spectrum. For the unbounded version of whole-plane SLE as studied in Refs. \cite{DNNZ,LoYe}, a phase transition has been shown to occur for high enough moments from the bulk  spectrum towards a novel spectrum related to the point at infinity. For the bounded version of whole-plane SLE of Ref.  \cite{BeSmSLE}, a similar transition phenomenon, now associated with the SLE origin, is proved to exist for low enough moments, but we show that it is superseded by the earlier occurrence of the transition to the SLE tip spectrum. 
 \end{abstract}
 \maketitle
\section{Introduction}

		Harmonic measure is one of the fundamental objects in geometric function theory and its fine structure provides much information about the underlying geometry. We refer the reader to the survey by Makarov \cite{Makarov} or to the recent monograph \cite{MR2450237}. In this article, we focus on the integral means spectrum of the harmonic measure (see definition below).

Schramm-Loewner Evolution (SLE)  curves are (in part conjecturally) the conformally invariant scaling limits of interfaces in critical lattice models of statistical physics. Since its introduction by the late Oded Schramm \cite{Schramm} fifteen years ago, the SLE process has sparked intense interest both in mathematical and physical communities. (See Ref. \cite{Lawler05} for a detailed study of SLE.)  One direction of research is the study of geometrical properties of these curves, one approach to which is precisely to compute the multifractal spectra for the harmonic measure. 

The first result in this direction is due to the second author who computed the SLE bulk spectrum by means of quantum gravity methods \cite{Duplantier00,MR2112128}, followed by Hastings who also computed a spectrum associated with the neighborhood of the SLE tip \cite{Hastings}. In Ref. \cite{BeSmSLE}, the first author and Smirnov provided a rigorous approach to the {\em average} integral means spectrum for whole-plane SLE, and showed that in a certain range of parameters (i.e., large enough negative moment orders), the tip spectrum dominates. 

An {\em unbounded} version of  whole-plane SLE was also studied in Refs. \cite{Hal-DNNZ,DNNZ} and \cite{IL,2012arXiv1203.2756L,LoYe} for which it was shown that for a large enough positive moment order, the (average) bulk integral mean spectrum undergoes a {\em phase transition} towards a novel  form, which was argued to be related to the point at infinity. For the {\em bounded} version of whole-plane SLE as studied in Ref. \cite{BeSmSLE},  which is related by inversion to the unbounded one, a similar transition phenomenon may thus occur near the origin, which can be interpreted as the starting point of the random curve (see Definition \ref{def:wplSLE} below). Indeed, it was observed in Refs. \cite{LoYe, DNNZ} that the analysis provided in Ref.  \cite{BeSmSLE} is  incomplete, and that the integral means spectrum could a priori be dominated by a novel spectrum, thought of as arising from the neighborhood of the starting point, and to be distinguished from that brought in by the vicinity of the tip. 

The purpose of this work is to complete the analysis undertaken in Ref. \cite{BeSmSLE}, so as to rigorously establish,  for the bounded version of whole-plane SLE as studied there, the form of  the averaged integral means spectrum for all moment orders. In particular, we show that for large enough negative moments, the new spectrum dominates the bulk one, but that both are still dominated by the tip spectrum in the same range of parameters. The results are summarized in Theorem \ref{thm:main}. The existence of the new spectrum is established starting in Section \ref{belowt1}. Following Ref. \cite{DNNZ}, the last Section \ref{Sectip} briefly discusses the relation of this spectrum to the derivative exponents of Ref. \cite{MR2002m:60159b} for standard radial SLE, or, equivalently, to the non-standard tip exponents of Ref. \cite{MR2112128}; it further heuristically suggests why the new spectrum should be associated with  the `second tip' of bounded whole-plane SLE, the image by inversion of the point at infinity in the unbounded version.

Before we proceed with the details of the analysis, let us also mention the work  by Johansson Viklund and Lawler \cite{0911.3983}, who established the almost sure version of the SLE tip multifractal spectrum, that by Alberts, Binder and Viklund \cite{Alberts2016} on the almost sure dimension spectrum for SLE boundary collisions, as well as the recent preprint  by Gwynne,  Miller, and Sun \cite{2014arXiv1412.8764G}, who used the so-called ``Imaginary Geometry'' of Miller and Sheffield \cite{ms2012imag1,ms2012imag2,ms2012imag3,ms2013imag4} to compute the a.s. value of the SLE bulk multifractal spectrum.

\subsection{Definitions and Statements} 
Let $\Omega=\C\setminus K$ where $K$ is a simply-connected compact set and let
$\phi$ be a Riemann mapping from $\D_-$ (i.e., the complement of the unit disk $\mathbb D$)  onto
$\Omega$ such that $\phi(\infty)=\infty$. {\em The integral means spectrum}
of $\phi$ (or $\Omega$) is defined as
$$
\beta_\phi(t)=\beta_\Omega(t)=\limsup_{r \to 1^+}
\frac{\log \int_0^{2\pi} |\phi'(re^{i\theta})|^td\theta}{-\log(r-1)}.
$$

For random fractals, it is natural to study {\em the average integral means spectrum}, which is defined as
$$
\bar\beta(t)=
\limsup_{r\to 1^+}\frac{\log \int_0^{2\pi} \E\brb{|\phi'(re^{i\theta})|^t}d\theta}{-\log|r-1|}.
$$
In this work, we are interested in the average integral means spectrum of whole-plane SLE curves. 
%\qqq{these curves are \em{external} in the language of Duplantier-Zinsmeister}
\begin{definition}\label{def:wplSLE}
Let $\xi(t)=\exp(i\sqrt{\kappa}B_t)$ be a two-sided Brownian motion on the unit circle with $t\in \mathbb R$ and $\kappa > 0$. The whole-plane SLE$_\kappa$ is the family of conformal maps $g_t$ satisfying
\begin{equation}
\label{eq:SLE}
\partial_t g_t(z)=g_t(z)\frac{\xi(t)+g_t(z)}{\xi(t)-g_t(z)},
\end{equation}
with initial condition
%\begin{equation}\label{eq:cond}
$$
\lim_{t\to-\infty}e^{t}g_t(z)=z, \qquad z\in \C\setminus \{0\}.
$$
%\end{equation}
\end{definition}
This map $g_t$ is a conformal map from $\mathbb C\setminus K_t$ onto ${\mathbb D}_{-}$, where the compact set $K_t$ is the so-called \emph{hull} of the SLE process, and it describes a family of hulls that grow from the origin towards infinity. This is the so-called {\em exterior} version of  whole-plane SLE, as the map $g_{0}$  may be seen as the limit of a rescaled version of a radial SLE process growing from the unit disc towards infinity \cite{BeSmSLE}.  The integral means spectrum of whole-plane SLE is then defined as that of the \emph{inverse map} $\phi=g_0^{-1}$. Another version  describes hulls growing from infinity towards the origin, and is called {\em interior} whole-plane SLE. The integral means spectrum of this \emph{unbounded} process is studied in Ref. \cite{DNNZ}.

\begin{remark}
The interior and exterior versions of whole-plane SLE are conjugate under the map $z\mapsto 1/z$, and  their integral means spectra can be unified in a single formalism by considering \emph {mixed moments}, involving powers of the moduli of the conformal map and of its derivative, as studied in Ref. \cite{2015arXiv150405570D}. %The study made in Ref. \cite{} can thus be seen as a generalization of the present work. 
\end{remark}
The main result of this paper is (Fig. \ref{fig})
\begin{theorem}
\label{thm:main}
i) The average integral means spectrum of the exterior whole-plane SLE$_\kappa$ is given by 
$$
\begin{aligned}
\beta_{\mathrm{tip}}(t)=-t-1+\frac{1}{4}\big(4+\kappa-\sqrt{(4+\kappa)^2-8\kappa t}\big), & \qquad t \le t_2, \\
\beta_0(t)=-t+\frac{4+\kappa}{4\kappa}\big(4+\kappa-\sqrt{(4+\kappa)^2-8\kappa t}\big),  & \qquad  t_2\le t \le t_3, \\
\beta_{\mathrm{lin}}(t)= t-\frac{(4+\kappa)^2}{16\kappa}, & \qquad t_3\le t.
\end{aligned}
$$
ii) If in the definition of the spectrum we integrate over any set that excludes the neighbourhood of $\theta=0$, which corresponds to excluding the influence of the tip of the curve, then the average spectrum is given by
$$
\begin{aligned}
\beta_1(t)=-t-\frac{1}{2}\big(1+\sqrt{1-2\kappa t}\big),& \qquad t \le t_1, \\
\beta_0(t)= -t+\frac{4+\kappa}{4\kappa} \big(4+\kappa-\sqrt{(4+\kappa)^2-8\kappa t}\big),& 
\qquad  t_1\leq t \le t_3, \\
 \beta_{\mathrm{lin}}(t)=t-\frac{(4+\kappa)^2}{16\kappa}, & \qquad t_3\le t.
\end{aligned}
$$
In the above, the transition values for the moment order $t$ are given by
$$
%\begin{aligned}
%&t_1=-\frac{1}{128}(4+\kappa)^2(8+\kappa),\\
%&t_2=-1-\frac{3\kappa}{8},\\
%&t_3=\frac{3(4+\kappa)^2}{32\kappa},
%\end{aligned}
%\begin{aligned}
t_1:=-\frac{1}{128}(4+\kappa)^2(8+\kappa),\quad
t_2:=-1-\frac{3\kappa}{8},\quad
t_3:=\frac{3(4+\kappa)^2}{32\kappa},
%\end{aligned}
$$
and such that $t_1 < t_2< t_3$.
\end{theorem}
\begin{figure}[htbp]
\begin{center}
\includegraphics[scale=0.7]{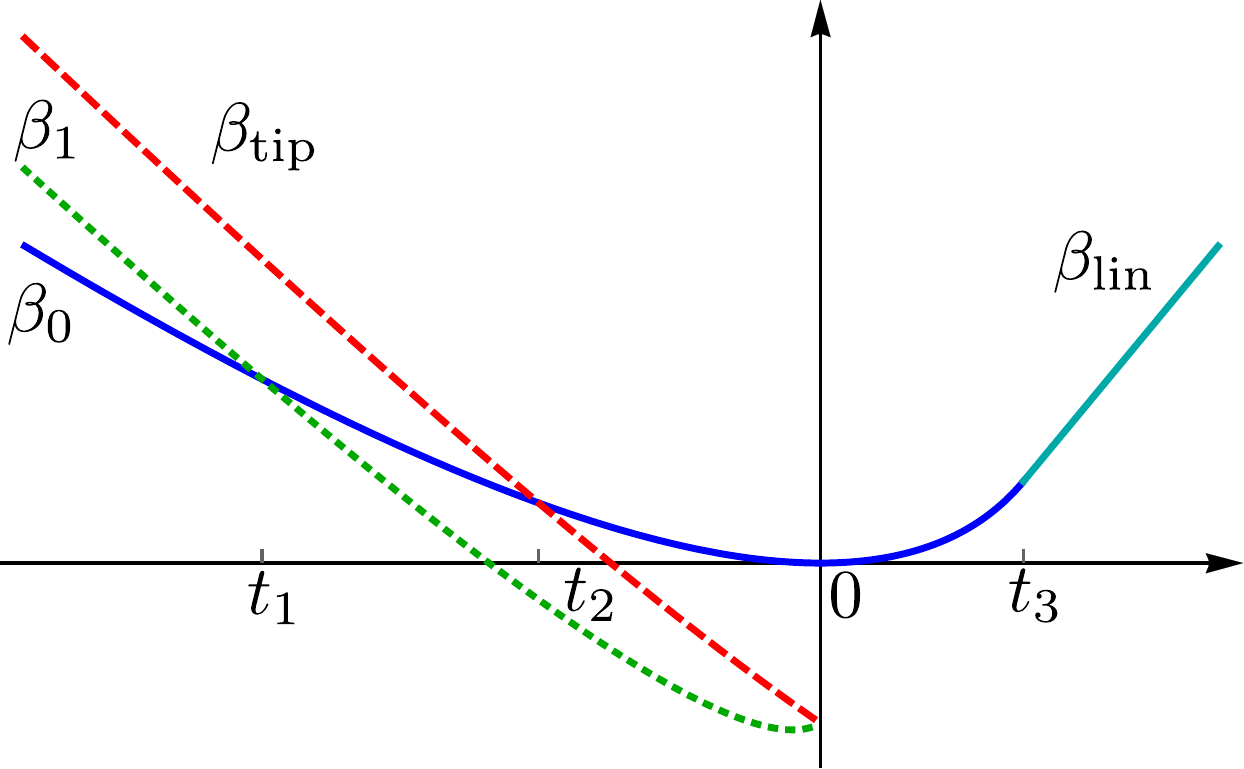}
\caption{Plot showing the relative positions of the various forms taken by the integral means spectrum in Theorem \ref{thm:main}: tip spectrum $\beta_{\mathrm{tip}}$, novel spectrum $\beta_1$, bulk spectrum $\beta_0$ and linear spectrum $\beta_{\mathrm{lin}}$. The $\beta_1$ spectrum is the analogue for the exterior whole-plane  SLE of the spectrum at infinity of Refs. \cite{DNNZ, LoYe} for the interior case (see also Ref. \cite{2015arXiv150405570D}). Note that $\beta_1$ supersedes the bulk spectrum below $t_1$, but stays below the tip spectrum.}
\label{fig}
\end{center}
\end{figure}
\begin{remark}
Point {\it i)} in Theorem \ref{thm:main} agrees with Ref. \cite{BeSmSLE}, while Statement {\it ii)} pertains to this work.
\end{remark}

\section{Differential operator}
As seen above, the integral means spectrum is defined for a map  from $\D_-$ onto some domain, i.e., in terms of the inverse of the SLE map. The Loewner equation for $g_t^{-1}$ is a PDE instead of an ODE and is much harder to work with. Instead, we study the {\em backward evolution}. It is well-known that for a standard radial [exterior] SLE process $\hat g_t$ (i.e., obeying Eq. \eqref{eq:SLE} for $t\geq 0$ 
with $\hat g_0(z)=z$), reversing time, i.e., changing sign in front of  Eq. \eqref{eq:SLE}, leads to solutions $f_t:=\hat g_{-t}, t\geq 0$, also called [backward] radial SLE,  that have, up to conjugation by rotation, the same distribution as $\hat g_t^{-1}$ (see Ref. \cite{BeSmSLE}, Lemma 1, which is an analog of Lemma 3.1 in Ref.  \cite{RoSch}). 

To compute the average integral means spectrum, let us then introduce the function
 $\tilde F(z;\tau):=\E\brb{|f'_\tau(z)|^t}$, where $f_\tau(z), \tau\geq 0,$ is such 
a backward radial SLE$_\kappa$ process. (Actually, this function depends also on the moment order $t$ and on $\kappa$, 
but these are fixed throughout the proof and we will no longer mention the dependence thereof.) The function $\tilde F$ is $C^{\infty,1}$  in $z,\tau$, since $f'_\tau(z)$ is by Loewner theory, and this remains true under expectation by standard dominated convergence and distortion theorems. As was shown in Ref. \cite{BeSmSLE}, it satisfies a parabolic PDE: 
%Following Ref. \cite{BeSmSLE}, we have 
\begin{lemma}
\label{lemma:martingale}
The function $\tilde F(z;\tau)$ is a solution to the PDE in variables $r,\theta,\tau$,
\begin{equation}
\label{eq:martingale0}
\begin{aligned}
t \frac{r^4+4 r^2(1-r \cos \theta)-1}{(r^2-2 r \cos \theta +1)^2}\tilde F+
\frac{r(r^2-1)}{r^2- 2 r \cos\theta+1}\tilde F_r 
\\
-\frac{2 r \sin\theta}{r^2-2 r \cos \theta +1} \tilde  F_\theta+
\frac{\kappa}{2} \tilde F_{\theta,\theta}-\tilde F_\tau=0,
\end{aligned}
\end{equation}
where $z=re^{i\theta}$.
\end{lemma}
Lemma 3 in Ref. \cite{BeSmSLE} further shows that there exists a limit to $e^{-\tau}f_\tau(z)$ as $\tau\to +\infty$, which has the same  distribution as $g_0^{-1}(z)$, where $g_t$ is a whole-plane SLE process as in Definition \ref{def:wplSLE};  hence we introduce
$$
F(z):=\lim_{\tau\to +\infty} e^{-\tau t}\tilde F(z;\tau),
$$
where limit and expectation commute by the same arguments as above. The average integral means spectrum of the exterior whole-plane SLE is thus associated with the singular behavior of 
$ \int_0^{2\pi} F(r e^{i\theta})d\theta$ as $r\to 1^+$.

Multiplying Eq. (\ref{eq:martingale0}) by $e^{-\tau t}$, and passing to the limit, shows that $F$ is a solution to the PDE in $r,\theta$,
\begin{equation}
\label{eq:martingale}
\begin{aligned}
t \br{\frac{r^4+4 r^2(1-r \cos \theta)-1}{(r^2-2 r \cos \theta +1)^2}-1} F+
\frac{r(r^2-1)}{r^2- 2 r \cos\theta+1} F_r
\\
-\frac{2 r \sin\theta}{r^2-2 r \cos \theta +1}   F_\theta+
\frac{\kappa}{2}  F_{\theta,\theta}=0.
\end{aligned}
\end{equation}
The exchange above, of the $\tau \to +\infty$ limit and of partial derivativation of $\tilde F(z;\tau)$ with respect to $r$ and $\theta$, is justified by the fact that the $\tau$-family, $e^{-\tau}f'_\tau(z)$, and all its $z$-derivatives are normal, i.e., uniformly bounded in any compact of $\mathbb D_-$, so that the spatial derivatives of $e^{-\tau t}\tilde F(z;\tau)$ form an equicontinuous family. A further requirement is that $\lim_{\tau \to +\infty} \frac{\partial}{\partial \tau} \big(e^{-\tau t} \tilde F\big)=0$. Use of the Schramm-Loewner equation \eqref{eq:SLE} for $f_\tau$ shows that
$$\frac{\partial}{\partial \tau} |e^{-\tau}f'_\tau(z)|^t=2t |e^{-\tau}f'_\tau(z)|^t \Re \frac{\xi(\tau)^2}{(f_\tau(z)-\xi(\tau))^2}.$$
Classical Koebe distortion theorems then show that the right-hand side is bounded by $C(z) e^{-2\tau}$, with $C$  defined on $\D_-$; this insures both the validity of the exchange of expectation and $\tau$-derivation, and the vanishing limit above. 
%in $\E\brb{f'_\tau(z)^{t/2}\overline{f'_\tau(z)}^{t/2}}$} 

It is easy to see that after a change of variables, this equation is still of \emph {parabolic} type, where $\theta$ plays the r\^ole of a spatial variable, and $r\to 1^+$ corresponds to time going to infinity \cite{BeSmSLE}.  

Instead of polar coordinates, it turns out to be more convenient to work with $(z$, $\bar{z})$ coordinates, where $F=F(z,\bar z)$ now formally depends on both complex variables, and  where Eq. \eqref{eq:martingale} becomes,
\begin{align}
\label{eq:operator}
&\Lambda F(z,\bar z)=0,\\ \nonumber
&\Lambda:=-\frac{\kappa}{2}\br{z\d -\zb\db}^2+\frac{z+1}{z-1}z\d +\frac{\zb+1}{\zb-1}\zb\db-
t\br{\frac{1}{(z-1)^2}+\frac{1}{(\zb-1)^2}}.
\end{align}
(See Ref. \cite{DNNZ} for details.)

We wish to study how the solutions to Eq. \eqref{eq:operator} behave when $z$ approaches the unit circle. In contradistinction to the interior case \cite{DNNZ}, it seems difficult to construct explicit solutions, so we are left with  constructing sub- and super-solutions with same boundary behaviors \cite{BeSmSLE}. 

Following Refs. \cite{BeSmSLE,DNNZ}, we consider the action of $\Lambda$ on functions of the peculiar form,
\begin{equation}\label{eq:psi}
\psi(z,\zb):=(z\zb-1)^{-\beta}g({u})=(|z|^2-1)^{-\beta}g(u),
\end{equation}
where $g$ is a $C^2$ function of ${u}:=(1-z)(1-\zb)=|1-z|^2$.
By looking at the leading terms in Eq. \eqref{eq:operator} for $\psi$ as $r\to 1^+$, one obtains a `boundary equation' for $g$ (see Refs \cite{BeSmSLE,DNNZ} for details),
\begin{equation}
\label{eq:boundary}
\br{t(2-{u})-2\beta}g({u})+\br{\frac{\kappa}{2}(2-{u})-(4-{u})}{u} g'({u})+\frac{\kappa}{2}(4-{u}){u}^2g''({u})=0;
\end{equation}
from now on, we assume that $g$ in \eqref{eq:psi} satisfies this equation. 

\begin{comment}
We shall also be interested in functions $g$ of the form $g(u)={u}^\gamma g_0({u})$, $\gamma\in \mathbb R$, and such that Eq. \eqref{eq:boundary} may reduce to a hypergeometric-type equation on $g_0$, as we shall see below \cite{BeSmSLE}. 

Let us now return to the action of the differential operator \eqref{eq:operator} on $\psi$ \eqref{eq:psi} in $\mathbb D_-$ and follow Ref. \cite{DNNZ}, Section (4.2.). Making use of Eq. \eqref{eq:boundary} to eliminate the second order derivative of $g$, we finally obtain for $\psi(z,\bar z)=(z\zb-1)^{-\beta}{u}^\gamma g_0({u})$,  
\begin{equation}
%\label{eq:action}
\begin{aligned}
%\begin{split}
\frac{\Lambda \psi}{\psi}&=
(1-z\zb)\brb{\frac{1}{{u}}(\beta-\gamma-t)+\frac{1}{4-{u}}\br{-\beta-t-\frac{\kappa}{2}\gamma}}\\
&+(1-z\zb)\brb{\br{\frac{\kappa}{2}-1-\frac{2\kappa}{4-{u}}}\frac{g'_0({u})}{g_0({u})}}\\
&+\frac{(1-z\zb)^2}{{u}^2}\brb{\frac{1}{4-{u}}\br{-2t-2\beta-\kappa\gamma-\kappa{u}\frac{g'_0({u})}{g_0({u})}}+\br{\frac{\kappa}{2}+1}\br{\gamma+{u}\frac{g'_0({u})}{g_0({u})}}}.
%\end{split}
\end{aligned}
\end{equation}
\end{comment}
Let us then consider the action of the differential operator \eqref{eq:operator} on $\psi$ \eqref{eq:psi} in $\mathbb D_-$ and follow Ref. \cite[Section 4.2.]{DNNZ}. Making use of Eq. \eqref{eq:boundary} to eliminate the second derivative of $g$, one obtains after some computation, 
\begin{equation}
\label{eq:action}
\begin{aligned}
%\begin{split}
\frac{\Lambda \psi}{\psi}&=
(z\zb-1)\brb{\frac{1}{{u}}(t-\beta)+\frac{1}{4-{u}}\br{\beta+t}- 
\br{\frac{\kappa}{2}-1-\frac{2\kappa}{4-{u}}}\frac{g'({u})}{g({u})}}\\
&+\frac{(z\zb-1)^2}{{u}^2}\brb{\frac{1}{4-{u}}\br{-2t-2\beta-\kappa{u}\frac{g'({u})}{g({u})}}+\br{\frac{\kappa}{2}+1}\br{{u}\frac{g'({u})}{g({u})}}}.
%\end{split}
\end{aligned}
\end{equation}

We are also interested in the action of $\Lambda$ on $\psi$ functions with logarithmic corrections \cite{BeSmSLE}. Let us introduce 
$$
\ell_{\delta}=\ell_{\delta}(z,\zb):=\br{-\log(z\zb-1)}^{\delta}.
$$
It is then easy to see that, 
\begin{equation}
\label{eq:action_log}
\frac{\Lambda(\psi \ell_{\delta})}{\psi \ell_{\delta}} =\frac{ \Lambda(\psi)}{\psi}-\frac{2\delta z\zb}{{u}(-\log(z\zb-1))}.
\end{equation}
For certain choices of exponent $\beta$ and of $g$, some leading terms in \eqref{eq:action} may cancel out so that the second term on the r.h.s. of Eq. \eqref{eq:action_log} dominates.  The latter has a sign opposite to the arbitrary sign of $\delta$, which means that if $\psi$ is positive, we shall be able construct sub- and super-solutions $\psi \ell_{\delta}$ with growth rate arbitrary close to that of $\psi$.

\section{Boundary solutions and their extension to $\mathbb D_-$}
\subsection{Hypergeometric equation}
We shall  be interested in functions $g$ of the form $g(u)={u}^\gamma g_0({u})$, $\gamma\in \mathbb R$, and such that Eq. \eqref{eq:boundary} may reduce to a hypergeometric-type equation on $g_0$ \cite{BeSmSLE}. 
Upon substituting this  into  \eqref{eq:boundary} and factoring out $u^\gamma$, we obtain,
\begin{align} \label{eq: eq for g0}
& \left(2\beta(\gamma)-2\beta +A(\gamma)u\right)g_0(u) \\   \nonumber &+\left(\frac{\kappa}{2}(2-u)+(\kappa\gamma-1)(4-u)\right)\,ug_0'(u)+\frac{\kappa}{2} (4-u)\, u^{2}g_0''(u)=0,
\end{align}
in terms of the quadratic polynomials,
\begin{align}
\label{eq:betagamma}
& \beta(\gamma):=\kappa\gamma^2-\left(\frac{\kappa}{2}+2\right)\gamma +t, \\
\label{eq:A}
&A(\gamma):=- \frac{\kappa}{2}\gamma^2+\gamma-t.%
%&C(\gamma):= -\frac{\kappa}{2}\gamma^2+\left(\frac{\kappa}{2}+2\right)\gamma-t,\\
%\kappa\gamma^2/2-C(\gamma).
\end{align} 
%such that 
%\begin{equation}\label{eq:AbC}
%A(\gamma)+\beta(\gamma)+C(\gamma)=\gamma-t.
%\end{equation}
%\textcolor{blue}{We thus have
%\begin{eqnarray}\nonumber
%\beta(\gamma)=\kappa\gamma^2/2-C(\gamma)=\kappa\gamma^2-(\kappa/2+2)\gamma +p%\\ \nonumber
%&&A^{\sigma}(\gamma)=A(\gamma)-(1+\sigma)p=- \frac{\kappa}{2}\gamma^2+\gamma -\sigma p.
%\end{eqnarray}}
For the choice of parameter $\beta=\beta(\gamma)$,  Eq. \eqref{eq: eq for g0} then reduces to an hypergeometric equation, 
\begin{eqnarray} \label{eq:hypergeom}  A(\gamma) g_0(u)   +\left[\frac{\kappa}{2}(2-u)+(\kappa\gamma-1)(4-u)\right] g_0'(u)+\frac{\kappa}{2} (4-u)u g_0''(u)=0,\end{eqnarray}  
which is the same as Eqs. \cite[(184)]{DNNZ} or \cite[(17)]{BeSmSLE}. Following either of these papers, we see that the general solution is, %up to constant factor,  
\begin{equation}\label{eq:g0}
g_0(u)=C_0\,\F(a,b,c,{u}/4)-C_0'\, ({u}/4)^{1/2-a-b}\F(a',b',c',{u}/4),
\end{equation} 
where
\begin{align}
\label{eq:abc}
&a=a(\gamma)=\gamma-\gamma_+, \quad b=b(\gamma)=\gamma-\gamma_-, \quad c=\frac{1}{2}+a+b,\\ \nonumber
&a'=\frac{1}{2}-a, \quad b'=\frac{1}{2}-b, \quad c'=\frac{1}{2}+a'+b',
\end{align}
and where $\gamma_{\pm}:=(1\pm\sqrt{1- 2\kappa t})/\kappa$ are the two roots in  
$A(\gamma)=-\frac{\kappa}{2}(\gamma-\gamma_+)(\gamma-\gamma_-)$.% is dual to $\tg$: it is the other solution of \eqref{eq:beta_tild_equation}. 

Hypergeometric functions are singular at ${u}=4$, but the solution should be smooth at $z=-1$, which means that the coefficients $C_0$ and $C_0'$ should be chosen in such a way that the singular parts cancel out \cite{BeSmSLE}. This precisely happens, up to a constant factor, for 
\begin{equation}
\label{eq:C0}
g_0(0)=C_0=\frac{\Gamma(3/2-a-b)}{\Gamma(1/2-a)\Gamma(1/2-b)},\,\,\,C_0'=\frac{\Gamma(c)}{\Gamma(a)\Gamma(b)},
\end{equation}
such that near $u=4$ (see Ref. \cite{BeSmSLE}),
\begin{equation}
\label{eq:g04}
g_0(u)= \frac{1}{\sqrt{\pi}}\left(\frac{1}{2}-a-b\right)+ O(4-u).
\end{equation}
\subsection{Action of $\Lambda$ on trial functions $\psi$}
Following Ref. \cite[Section 4.2.3]{DNNZ}, let us now return to the action \eqref{eq:action}, this time for $\psi(z,\bar z)=(z\zb-1)^{-\beta}{u}^\gamma g_0({u})$, and for the choice $\beta=\beta(\gamma)$ \eqref{eq:betagamma}.  Using $ug'/g=ug'_0/g_0+\gamma$,  
we finally obtain, 
\begin{align} \nonumber\frac{\Lambda\psi}{\psi}=&(z\bar z-1)\left[\frac{1}{u}\big[t+\gamma-\beta(\gamma)\big]-\frac{2A(\gamma)}{4-u}-\left(\frac{\kappa}{2}-1-\frac{2\kappa}{4-u}\right)\frac{g'_0(u)}{g_0(u)}\right] \\ \nonumber &+\frac{(z\bar z-1)^2}{u^2} \left[\frac{1}{4-u}\left(4A(\gamma)-\kappa u\frac{g_0'(u)}{g_0(u)}\right)+\left(1+\frac{\kappa}{2}\right)\left(\gamma+u\frac{g_0'(u)}{g_0(u)}\right)\right].\\ 
&& \label{eq:actionbis}%\\&&(\kappa \gamma^2-2\beta)+C(\gamma)\left[\frac{1-z\bar z}{x}\left(\frac{1-z\bar z}{x}+1\right)-\frac{2}{x}\right]\\\label{genrester}&&-A(\gamma) \left(\frac{1-z\bar z}{x}-1\right)-(1+\sigma)p,
\end{align}
In Eq. \eqref{eq:actionbis}, there seems to be an unexpected singularity at ${u}=4$ ($z=-1$), but in fact the choice of constants $C_0$ and $C_0'$ in Eq. \eqref{eq:C0}, made to ensure that $g_0(u)$ is regular at $u=4$,  further yields
$$
\frac{g_0'(u)}{g_0(u)}=-\frac{1}{2}ab+O(4-{u})=\frac{1}{\kappa}A(\gamma)+O(4-{u}),
$$
so that the singularity at ${u}=4$ cancels out in the action \eqref{eq:actionbis}.

In the $u \to 0$ limit, Eq. \eqref{eq:g0} shows that (up to a non-vanishing coefficient)
\begin{equation}
\label{eq:alpha}
u\frac{g_0'(u)}{g_0(u)}\sim u^{\alpha},\,\,\, \alpha:=\min\left\{\frac{1}{2}-a-b,1\right\},
\end{equation}
\begin{remark}\label{remark2}
In the case where $C_0$ vanishes,  $\alpha=1/2-a-b$.
\end{remark}
The second line in Eq. \eqref{eq:actionbis} is in this limit, 
%\begin{align}\label{Cuto0}&&\frac{(1-z\bar z)^2}{x^2} \left\{A^\sigma(\gamma)+(\sigma-1)p+\left(\frac{\kappa}{2}+1\right)\gamma\right\}\\ \nonumber&&=\frac{(1-z\bar z)^2}{x^2} \left\{A(\gamma)-2p+\left(\frac{\kappa}{2}+1\right)\gamma\right\}  = \frac{(1-z\bar z)^2}{x^2} C(\gamma).\end{align}
\begin{align}\label{eq:Cuto0}%&\frac{(1-z\bar z)^2}{u^2} \left\{A(\gamma)+\left(\frac{\kappa}{2}+1\right)\gamma\right\}\\ \nonumber& = 
\frac{(z\bar z-1)^2}{u^2} \left[C(\gamma)+O(u^{\alpha})\right],\end{align}
with 
\begin{align}\label{eq:C}C(\gamma):=A(\gamma)+\left(1+\frac{\kappa}{2}\right)\gamma=-\frac{\kappa}{2}\gamma^2+\left(2+\frac{\kappa}{2}\right)\gamma-t.
%\\ \nonumber& = \frac{(1-z\bar z)^2}{u^2} \left[C(\gamma)+O(u^{1/2-a-b})\right],
\end{align}

\subsection{Beliaev-Smirnov solution}\label{BS}
The first pair of exponents $\beta$ and $\gamma$ that we are interested in have been introduced in Ref. \cite{BeSmSLE}, so as to cancel the leading singularity in \eqref{eq:Cuto0}. 
%Let us define the quadratic forms \cite{DNNZ} 
%\begin{equation}\label{eq:Cbeta}
%\begin{cases}
%C(\gamma):=-\frac{\kappa}{2}\gamma^2+(2+\frac{\kappa}{2})\gamma-t,\\
%\beta(\gamma):=\frac{\kappa}{2}\gamma^2-C(\gamma)=t-(2+\frac{\kappa}{2}) \gamma+\kappa\gamma^2.
%\end{cases}
%\end{equation}
For a given $t$, there are two solutions to $C(\gamma_0)=0$, and we consider the particular values,
\begin{equation}
\label{eq:beta_not}
\begin{cases}
\gamma_0=\frac{1}{2\kappa}(4+\kappa-\sqrt{(4+\kappa)^2-8\kappa t}),\\
\beta_0=\beta(\gamma_0)=\frac{\kappa}{2}\gamma_0^2=-t+(2+\frac{\kappa}{2})\gamma_0=-t+\frac{4+\kappa}{4\kappa}(4+\kappa-\sqrt{(4+\kappa)^2-8\kappa t}).
\end{cases}
\end{equation}

%These particular values are denoted by $\gamma_0$ and $\beta_0$. 
\begin{comment}
If we plug in $\beta_0$ into \eqref{eq:boundary} then we have
\begin{equation}
-\frac{\kappa+2}{2}\gamma_0 g_0({u})+\br{\frac{\kappa}{2}(2-{u})+(\kappa\gamma_0-1)(4-{u})}g_0'({u})+
\frac{\kappa}{2}(4-{u}){u} g_0''({u})=0
\end{equation}
This is the same as equation \cite[(185)]{DNNZ} or \cite[(17)]{BeSmSLE}. Following either of these papers we see that the general solution is proportional to 
$$
\F(a,b,c,{u}/4)-C_0 ({u}/4)^{1/2-a-b}\F(a',b',c',{u}/4)
$$ 
where
\begin{align}
a=\gamma_0-\tg, \quad b=\gamma_0-\tg', \quad c=\frac{1}{2}+a+b\\
a'=\frac{1}{2}-a, \quad b=\frac{1}{2}-b, \quad c'=\frac{1}{2}+a'+b'
\end{align}
and $\tg'=(1-\sqrt{1-2\kappa t})/\kappa$.% is dual to $\tg$: it is the other solution of \eqref{eq:beta_tild_equation}. 

Hypergeometric functions are singular at ${u}=4$ but the solution should be smooth at $z=-1$, this means that the coefficient $C_0$ should be chosen in such a way that the singular parts cancel out. This happens for
\begin{equation}
C_0=\frac{\Gamma(c)}{\Gamma(a)\Gamma(b)}\frac{\Gamma(a')\Gamma(b')}{\Gamma(c')}.
\end{equation}
\end{comment}

The first condition is for the boundary solution $g_0$ \eqref{eq:g0} to be bounded (i.e., ${u}^{\gamma_0}$ should be the only singular term in $g$), which means that $1/2-a-b\ge 0$, where now $a=a(\gamma_0), b=b(\gamma_0)$ as in Eqs.  \eqref{eq:abc}. Simple algebra shows that this is equivalent to $t \le t_3$. 

The second condition is that $g_0(u)$ should be positive for the whole range $u\in[0,4]$ when $z$ describes the unit circle. It was observed in Ref. \cite{BeSmSLE} that this happens when $1/2-b>0$. It was erroneously stated there that this is always true, while in fact this holds only for $t_1<t$. (See Proposition \ref{pr:b} below.) 

For $t\in(t_1,t_3)$ we thus have that $g_0$ is bounded and positive, and the arguments in Ref. \cite{BeSmSLE}, which we detail and refine here, stay valid. To study the spectrum, one has to analyze the behavior of $\Lambda \psi_0$,  where $\psi_0(z,\bar z)=(z\zb-1)^{-\beta_0}{u}^{\gamma_0}g_0({u})$ for values $\beta_0, \gamma_0$  as in Eq. \eqref{eq:beta_not}, and the associated hypergeometric function $g_0$ \eqref{eq:g0}. This is provided by Eq. \eqref{eq:actionbis}, where \eqref{eq:C} and  \eqref{eq:beta_not} give the explicit coefficients, $t+\gamma_0-\beta(\gamma_0)=2t-\left(1+\frac{\kappa}{2}\right)\gamma_0$, and $A(\gamma_0)=-\left(1+\frac{\kappa}{2}\right)\gamma_0$. 
\begin{comment}
\textcolor{blue}{\begin{equation}
%\nonumber\label{eq:action_not}
\begin{aligned}
(z\zb-1)\brb{
\frac{1}{{u}}\br{2t-\frac{\kappa+2}{2}\gamma}
+\frac{2}{4-{u}}\br{\frac{\kappa+2}{2}\gamma+\kappa\frac{g_0'}{g_0}}
-\frac{\kappa-2}{2}\frac{g_0'}{g_0}
}\\
+\frac{(z\zb-1)^2}{{u}^2}\brb{
\frac{1}{4-{u}}\br{-2(\kappa+2)\gamma-\kappa{u}\frac{g_0'}{g_0}}
+\frac{\kappa+2}{2}\br{\gamma+{u}\frac{g'_0}{g_0}}
}
\end{aligned}
\end{equation}}
\end{comment}
\begin{comment}
It seems that there might be an unexpected singularity at ${u}=4$ ($z=-1$), but in fact out choice of the constant $C$ ensures that this singularity cancels out. Indeed, for this $C$ we have that 
$$
\frac{g_0'}{g_0}=-\frac{\kappa+2}{2\kappa}\gamma+O(4-{u})
$$
and the singularity at ${u}=4$ cancels out.
\end{comment}

As a preparation also for a complete analysis below, let us detail various radial limits  as $|z|\to 1$ in action \eqref{eq:actionbis}, while recalling the geometrical constraint,
$|z|-1\leq |z-1|=u^{1/2}$.\\  
{\em Generic case:} $z\zb-1\to 0$, but ${u}$ is bounded away from $0$. In this case, 
$\Lambda \psi_0/\psi_0=O(z\zb-1)$.\\
{\em Special case:} $z$ approaches $1$, $u\to 0$, so that $z\zb-1=O({u}^{1/2})$. In this case $g_0'/g_0=O({u}^{\alpha}/{u})$ in \eqref{eq:alpha} is dominated by $1/u$, as we have shown before that $1/2-a-b > 0$ for $t < t_3$, so that $\alpha>0$. Thus the first line in \eqref {eq:actionbis} is of order $O({u}^{-1/2})$. The second line, as given by Eq. \eqref{eq:Cuto0}, seems to be of  order $O({u}^{-1})$, but because of the very choice of $\gamma=\gamma_0$ such that $C(\gamma_0)=0$, 
% if we examine the coefficients carefully we see that near ${u}=0$ it behaves as 
%$$
%\frac{{u}}{{u}^2}\brb{\frac{1}{4}\br{-2(\kappa+2)\gamma+O({u}^{1/2-a-b})}+\frac{\kappa+2}{2}\gamma+O({u}^{1/2-a-b})}.
%$$
%The main (constant) terms inside square brackets cancel out and the result is $O({u}^{1/2-a-b})$. Thus
 it is of order $O({u}^{\alpha-1})$, with $\alpha>0$ for $t<t_3$.\\  The results of this discussion can simply be recast as, 
 \begin{equation}
\label{eq:action_not}
\frac{\Lambda \psi_0}{\psi_0}=\frac{z\bar z-1}{{u}}O(u^0)+\frac{(z\bar z-1)^2}{{u}^2}O(u^{\alpha}),\,\,\, |z|-1\leq u^{1/2}.
\end{equation} 
To complete the proof of Theorem \ref{thm:main} in the range $t\in(t_1,t_3)$, one then considers as in Ref. \cite{BeSmSLE},  the set of logarithmic modifications of $\psi_0$ for all $\delta$,
$$
\ell_\delta\psi_0(z,\bar z)=(-\log(z\zb-1))^{\delta}(z\zb-1)^{-\beta_0}{u}^{\gamma_0}g_0({u}).
$$
Recalling the action \eqref{eq:action_log}, one sees  that for $|z|$ close enough to $1$ and for all $u$, both terms in   $\Lambda \psi_0/\psi_0$ \eqref{eq:action_not}  are dominated by  the logarithmic second term on the r.h.s. of \eqref{eq:action_log}:
\begin{equation}\label{eq:log}
\left|\frac{\Lambda \psi_0}{\psi_0}\right|\leq \frac{2|\delta| z\zb}{{u}(-\log(z\zb-1))}.
\end{equation} 
The positive function $\psi_0\ell_\delta$ is thus a  sub-solution for $\delta>0$ or a super-solution for $\delta<0$, so that the integral means spectrum is $\beta_0(t)$ for $t\in [t_2,t_3)$, or $\beta_{\mathrm{tip}}(t)=\beta_0(t)-2\gamma_0(t) -1$, for $t\in(t_1, t_2)$ where $2\gamma_0(t) +1<0$.  

By H\"older's inequality, the spectrum is convex, and  by standard distortion theorems, bounded by $t$ for $t>0$. This, together with the fact that $\partial_t \beta_0(t_3)=1$,  establishes Theorem \ref{thm:main} for $t\ge t_3$. 

\subsection{Below $t_1$}\label{belowt1}
At $t_1$, $b=1/2$, and  $g_0(0)=C_0=0$ in Eq. \eqref{eq:C0}. Actually, this vanishing also happens for higher half-integer values  $b=n+1/2$.  Proposition \ref{pr:tkn} below shows that there indeed exist a finite set of integers $ \mathcal J_\kappa$, and a finite discrete set $\mathcal T_\kappa$ of moment orders,  
%\begin{align}  \label{eq:tk}
%& \mathcal T_\kappa:=\{t_{1-n}, n\in \{0,\cdots,\lfloor \kappa^{-1}\rfloor\}\}\\ \label{eq:t1n}
%& t_{1-n}=t_{1-n}(\kappa)=-\frac{(1+2n)(8+\kappa-2n\kappa)(4+\kappa+2n\kappa)(4+\kappa-2n\kappa)}{128(1-n\kappa)^2},%\\ \label{eq:bn}
%& b=b(\gamma_0(t_{1-n}))=n+\frac{1}{2},\,\,\,0\leq n\leq \lfloor\kappa^{-1}\rfloor,
%\end{align}
\begin{align} \label{eq:Jk}
& \mathcal J_\kappa:=\{n\in \mathbb N, 0\leq n\leq \lfloor\kappa^{-1}\rfloor\}\\
 \label{eq:tk}
& \mathcal T_\kappa:=\{t_{1-n}, n\in \mathcal J_\kappa\}\\ \label{eq:tkn}
& t_{1-n}=t_{1-n}(\kappa):=-\frac{(1+2n)(8+\kappa-2n\kappa)(4+\kappa+2n\kappa)(4+\kappa-2n\kappa)}{128(1-n\kappa)^2},%\\ \label{eq:bn}
%& b=b(\gamma_0(t_{1-n}))=n+\frac{1}{2},\,\,\,0\leq n\leq \lfloor\kappa^{-1}\rfloor,
\end{align}
such that,  
%\begin{align}  \label{eq:bn}
%& b=b(\gamma_0(t_{1-n}))=n+\frac{1}{2},\,\,\,n\in \mathcal J_\kappa,%0\leq n\leq \lfloor\kappa^{-1}\rfloor,
%\\ \label{eq:g0tk}
%&g_0(0)=C_0=0,\,\,\,t\in \mathcal T_\kappa.
%\end{align} 
\begin{align} \label{eq:g0tk} b=b(\gamma_0(t_{1-n}))=n+\frac{1}{2},\,\,\,n\in \mathcal J_\kappa,\,\,\,
g_0(0)=C_0=0,\,\,\,t\in \mathcal T_\kappa.
\end{align}The $n=0$ case corresponds, for any value of $\kappa$, to the point $t_1$ as above; note also that  for $\kappa >1$, $\mathcal T_\kappa=\{t_1\}$, whereas strictly positive values of $n$ exist in $\mathcal J_\kappa$ only for $0<\kappa\leq 1$. 

In Eq. \eqref{eq:C0}, observe now that $1/2-a>0$, because $a=a(\gamma_0)<0$ in Eq. \eqref{eq:abc}, and recall that $1/2-a-b >0$ for $t<t_3$. Therefore, the sign of $g_0(0)$ \eqref{eq:C0} is given by that of $\Gamma(1/2-b)$, and by the very property of analytical continuation of the $\Gamma$-function, is thus alternating in the successive intervals $t\in(t_{-n},t_{1-n})$, being positive or negative for $n$ odd or even, respectively.  For later convenience, let us then introduce,
\begin{equation}\label{eq:sgng0}
\sigma=\sigma(t):=\sign g_0(0)=(-1)^{n-1},\,\,\, t\in(t_{-n},t_{1-n}),\,\,\,n\in \mathcal J_\kappa.
\end{equation} 
Notice that owing to Eq. \eqref{eq:g04}, $g_0(4)$ is always {\it positive} for $t<t_3$. Then, in the interval of moment orders, $t\in(t_{-n},t_{1-n})$, with $n\in \mathcal J_\kappa$, the graph of $g_0(u)$ possesses exactly $n+1$ simple zeroes over the interval $u\in (0,4)$. 
 %On the other hand, $1/2-a-b>0$ and hence it is $o(1/{u})$.
\subsection{Power-law solution}
Note that the hypergeometric equation \eqref{eq:hypergeom} becomes degenerate when $A(\gamma)=0$, with $g_0$ a constant solution and $g$ of the form ${u}^\gamma$. %This corresponds in Eq.  Indeed if we plug in $g={u}^\gamma$ into \eqref{eq:boundary} we see that this is a solution if $\beta$ and $\gamma$ satisfy the following system of equations
%\begin{equation}
%\label{eq:beta_tild_equation}
%\begin{cases}
%\beta=\beta(\gamma)=t-(2+\frac{\kappa}{2}) \gamma+\kappa\gamma^2\\
%\frac{\kappa}{2}\gamma^2-\gamma+t=0
%\end{cases}
%\end{equation}
As before, there are two solutions, $\gamma_{\pm}$, to this system, which correspond to the degenerate cases $a=0$ or $b=0$ in Eqs. \eqref{eq:g0} and \eqref{eq:abc}, and to $C_0=0$ in Eq. \eqref{eq:C0}. As in the case of interior whole-plane SLE \cite{DNNZ}, we are especially interested in the pair,% hereafter denoted by $\tb$ and $\tg$,
\begin{equation}
\label{eq:beta_tild}
\begin{cases}
\tg:=\gamma_+=\frac{1}{\kappa}(1+\sqrt{1-2\kappa t}),\\
 \tb=\beta(\tg)=-t-\frac{\kappa}{2}\tg=-t-\frac{1}{2}(1+\sqrt{1-2\kappa t}).
\end{cases}
\end{equation}
\begin{comment}
\begin{equation}
%\label{eq:beta_tild}
\begin{cases}
\tg=\frac{1}{\kappa}(1+\sqrt{1-2\kappa t}),\\
\tilde \beta=\beta(\tg)=-t-\frac{\kappa}{2}\tg=-t-\frac{1}{2}(1+\sqrt{1-2\kappa t}).
\end{cases}
\end{equation}
\end{comment}
\begin{remark}\label{rk:betagamma}
Both  $\beta_0$  \eqref{eq:beta_not} and $\tb$  \eqref{eq:beta_tild} are given by the same quadratic function $\beta(\gamma)$ \eqref{eq:betagamma}, in terms of $\gamma_0$ and $\tg$, respectively. 
\end{remark}
When plugging $\psi_1=\tp(z,\bar z):=(z\zb-1)^{-\tb}{u}^\tg$ into Eq. \eqref{eq:actionbis}, because $A(\tg)=0$ and $g'_0=0$,  many terms disappear, and the result is simply, 
\begin{equation}
\label{eq:action_tilde}
\frac{\Lambda \tp}{\tp}=\frac{z\zb-1}{{u}}\br{2t+\br{1+\frac{\kappa}{2}}\tg}+\frac{(z\zb-1)^2}{{u}^2}\br{1+\frac{\kappa}{2}}\tg.
\end{equation} 
As before, we distinguish two cases: $|z|\to 1$, but ${u}$ is bounded away from $0$;  ${u}\to 0$, and $|z|-1=O({u}^{1/2})$. 
In the first case, $\Lambda\tp/\tp$ is $O(|z|-1)$; in the second case, it is $O({u}^{-1})$.

\section{Mixing the two solutions}
For $t<t_1$, $g_0(u)$ changes sign at least once over the interval $(0,4)$, invaliding the proof of Section \ref{BS}. Recall that at the origin its sign alternates, as described in Eq. \eqref{eq:sgng0}.  The idea is thus to try and combine the two functions $\psi_0$ and $\psi_1$ into 
\begin{equation}\label{eq:psifin}
\psi:=\sigma \psi_0+\tp= \sigma g_0(u)u^{\gamma_0} (z\bar z-1)^{-\beta_0} + u^{\gamma_1}(z\bar z-1)^{-\beta_1},
\end{equation}
where $\sigma$ is as in Eq. \eqref{eq:sgng0}, so as to restore overall positivity for $\psi$. Then, in the action $\Lambda (\psi \ell_{\delta})$, the differential operator $\Lambda$ will act differently on $\psi_0$ and $\psi_1$, still maintaining the possibility to build sub- and super-solutions in this way. In this section, we are mostly interested in $t<t_1$, but some arguments are independent of that assumption, provided one assumes that, e.g., $t<0$, so that both $\psi_0$ and $\psi_1$ are defined. 
\begin{lemma}
\label{lemma:positive}
There is $r_0>1$ such that $\psi=\sigma \psi_0+\psi_1>0$ for all $z$ such that $1< |z|< r_0$.
\end{lemma} 
\begin{proof}
First of all, note that both $\psi_0$ and $\tp$ are continuous in the complement of the unit disc, and that $\tp>0$ everywhere. For $t>t_1$ we have $g_0>0$, hence both terms in $\psi$ are positive.   For $t=t_1$, $g_0(u)>0$ for $u>0$, so that $\psi>0$. For $t < t_1$, this is no longer true, but $\sigma g_0(u)\geq 0$ in some neighborhood $u\in [0,u_0]$ of $z=1$, hence $\psi >0$  there. Outside of this neighborhood, ${u}$ is uniformly bounded away from $0$, and also bounded above by $(|z|+1)^2$; therefore, for any $r_0>1$, there exist positive constants $c_0$, $c_1$ such that for all $z$ such that $1<|z|<r_0$ and $u=|z-1|^2>u_0$,
$$
\begin{aligned}
0\le & |\psi_0| \le  {c_0}(z\zb-1)^{-\beta_0},\\
\frac{1}{c_1}(z\zb-1)^{-\tb} < & \tp  <c_1(z\zb-1)^{-\tb}.
\end{aligned}
$$
For $t <t_1$, we have that $\tb>\beta_0$ (See Prop. \ref{pr:tb} below and Fig. \ref{fig}), so that for $|z|$ sufficiently close to $1$,  we have $\tp>|\psi_0|$,  hence $\psi>0$. 
\end{proof}

\begin{lemma}
\label{lemma:sign}
For  $t<t_1$ and $t\notin \mathcal T_\kappa$, there is $r_0>1$ such that  $\Lambda (\psi \ell_{\delta})$ for $\psi$ \eqref{eq:psifin} has a constant sign  in the annulus $1< |z|< r_0$, which depends only on that of $\delta$.
\end{lemma} 

\begin{proof}
As shown by Eq. \eqref{eq:action_log},  multiplying $\psi$ by a logarithmic factor $\ell_\delta$ results in an additional term in the action of the differential operator, whose sign depends on that of $\delta$ only. We shall show that near the boundary of the unit disc, this additional term is the main one, hence  $\Lambda (\psi \ell_{\delta})$ has constant sign there. Since $\psi$ is positive, this implies that $\psi \ell_{\delta}$ is a sub- or super-solution for $\delta >0$ or $\delta<0$, respectively. 

%\textcolor{blue}{Since $\gamma_0<0$ (for $t<0$), we see that $a+b=2\gamma_0-2/\kappa<0$, hence the order of the \textcolor{red}{first line} is the leading one in Eq. \eqref{eq:actionbis}, so that in this special case, $\Lambda \psi_0/\psi_0=O({u}^{-1/2})$.} 
As was shown in \eqref{eq:action_log},  \eqref{eq:action_not},  \eqref{eq:log} and \eqref{eq:action_tilde}, $\Lambda (\psi \ell_{\delta})/\ell_{\delta}$ can be written (up to smaller order terms) as the sum of four terms
\begin{equation}\label{eq:4mousquetaires}
\begin{aligned}
%&\psi_0\frac{r-1}{{u}}+\psi_0\frac{(r-1)^2}{{u}^2}{{u}^{1/2-a-b}}
\frac{\Lambda (\psi \ell_{\delta})}{\ell_{\delta}}=\sigma \psi_0\frac{-2\delta r^2}{{u}(-\log(r-1))}+\tp\frac{r-1}{{u}}+\tp\frac{(r-1)^2}{{u}^2}+\tp\frac{-2\delta r^2}{{u}(-\log(r-1))},
\end{aligned}
\end{equation}
where $r=|z|$.
We denote these terms by $\I$ -- $\IV$. 
Here, we omitted the positive constants in front of $\II, \III$, but terms $\I$ and $\IV$ are written in their {\em complete} form. For functions $A$ and $B$ of $r$, we shall use the short-hand notations, $A \gtrsim B$ for $A \geq c B$ with $c$ some positive constant, and $A \approx B$ when both $A\gtrsim B$ and $A \lesssim B$ hold. 
 
 Below, we consider three cases, and show that in each case one of the logarithmic terms (either $\I$ or $\IV$) is the leading one when $r\to 1^+$.

{\em Case 1:} ${u}$ is bounded away from zero. It is obvious that  $\IV \gtrsim \II, \III$ as $(r-1)\to 0$. For ${u}$ bounded away from zero, we have $\psi_0\approx (r-1)^{-\beta_0}$ and $\tp \approx (r-1)^{-\tb}$. Since $\tb>\beta_0$ for $t<t_1$ (Proposition \ref{pr:tb}), we have $\IV \gtrsim \I$. Also notice that $\tp$ is positive, hence for sufficiently small $r-1$, the sign of $\Lambda(\psi\ell_{\delta} )$ is opposite to that of $\delta$.  

{\em Case 2:} We assume that $(r-1)^{2-\varepsilon}<{u}<u_0$, where $\varepsilon>0$    and where, as above, $u_0>0$ is chosen such that $\sigma g_0({u})>0$ for $0<{u}<u_0$. Then,%(Note that by construction, $g_0(0)=1$ and hence such $c$ does  exist. In this case we have 
$$
\begin{aligned}
\II&= \psi_1\frac{r-1}{{u}},\\
\III&\lesssim \psi_1\frac{(r-1)^2}{{u}}\frac{1}{(r-1)^{2-\varepsilon}}= \psi_1 \frac{(r-1)^\varepsilon}{{u}},\\
\IV &\approx \psi_1 \frac{-\delta}{{u} (-\log(r-1))},
\end{aligned}
$$
and obviously $\IV$ dominates $\II$ and $\III$. Since $\sigma \psi_0$ is positive for our choice of $u_0$, we have that  both main terms $\I$ and $\IV$ in \eqref{eq:4mousquetaires} have a sign opposite to that of $\delta$.

{\em Case 3:} We assume that  $r-1>{u}^{1/2+\varepsilon}$ for some $\varepsilon>0$ that will be determined later. We also recall that $r-1\leq {u}^{1/2}$.  In this case  $-\log(r-1)\approx -\log {u}$.
First, we notice that
\begin{comment}
$$
\begin{aligned}
\I&\le \psi_0 \frac{{u}^{1/2}}{{u}},\\
\II&\le \psi_0 \frac{{u}^{1/2-a-b}}{{u}}, \\
\III&\approx \psi_0 \frac{1}{{u} (-\log {u})}
\end{aligned}
$$
and hence $\III$ is the main term out of these three. 
\end{comment}
$$
\begin{aligned}
\I&\approx \psi_0 \frac{-\delta}{{u} (-\log {u})},\\
\II&\lesssim \tp \frac{{u}^{1/2}}{{u}},\\
\III&\lesssim \tp \frac{1}{{u}}, \\
\IV&\approx \tp\frac{-\delta}{{u} (-\log {u})},
\end{aligned}
$$
so that both $\III$ and $\IV$ are dominated by $\III ':=\psi_1/u$. We would like to show that $\III ' \lesssim \I$, and  this requests comparing $\tp$ to $\psi_0$.  Proposition \ref{pr:t_tip} below precisely gives that   for $t<t_1$, one has
$\beta_{\mathrm{tip}}=\beta_0-2\gamma_0-1>\tb$, hence
$$
(r-1)^{-\tb}<(r-1)^{-\beta_0}(r-1)^{2\gamma_0+1},
$$
and
$$
\tp\approx (r-1)^{-\tb}{u}^{\tg}<(r-1)^{-\beta_0}{u}^{\gamma_0}{u}^{\tg-\gamma_0}(r-1)^{2\gamma_0+1}\approx \sigma \psi_0 {u}^{\tg-\gamma_0}(r-1)^{2\gamma_0+1}.
$$
Note that the last estimate requires that $g_0(0)\neq 0$, hence the condition $t\notin \mathcal T_\kappa$ in Lemma \ref{lemma:sign}. We expect ${u}^{\tg-\gamma_0}(r-1)^{2\gamma_0+1}$ to be bounded by some positive power of ${u}$. 
As was shown in Ref. \cite{BeSmSLE}, the definition of the threshold $t_2$ for the tip relevance is that $2\gamma_0+1<0$ for $t<t_2$, and we have here $t<t_1<t_2$ . Hence, for $r-1>{u}^{1/2+\varepsilon}$,
$$
{u}^{\tg-\gamma_0}(r-1)^{2\gamma_0+1}<{u}^{\tg-\gamma_0+(1/2+\varepsilon)(2\gamma_0+1)}={u}^{\tg+1/2 +\varepsilon(2\gamma_0+1)}.
$$
%It is shown in the proof of Proposition \ref{pr:b} that $\gamma_0+\tg>(4+\kappa)/2\kappa>0$ for $t<t_1$, hence 
%$$
%{u}^{\tg+\gamma_0+1}{u}^{\varepsilon(2\gamma_0+1)}<{u}^{1+\varepsilon(2\gamma_0+1)}.
%$$
Since $\gamma_1>0$,  for sufficiently small $\varepsilon$, the latter is bounded by some positive power of ${u}$, e.g., by $u^{1/2}$. Hence for this $\varepsilon$, we have that
\begin{equation}\label{eq:psi0psi1}
\tp\lesssim \sigma \psi_0{u}^{1/2}.
\end{equation}
 Thus we see that $\III '\lesssim \I$, and that the sign of  $\Lambda(\psi\ell_{\delta})$   is given by that of $\I$ in Eq. \eqref{eq:4mousquetaires}, which is opposite to the sign of $\delta$.

Altogether, these three cases show that in some neighborhood of the unit circle, the sign of $\Lambda(\psi\ell_{\delta})$ is constant and opposite to that of $\delta$.
\end{proof}
Lemmas \ref{lemma:positive} and \ref{lemma:sign} together show that $\psi \ell_{\delta}$ is a sub- or super-solution, depending on the sign of $\delta$. Following Ref. \cite{BeSmSLE}, one obtains that for $F$ such that $\Lambda F=0$, there exist positive constants $c_2$ and $c_3$, such that in some annulus adherent to $\mathbb D$, $c_2 \psi \ell_{-\delta}<F<c_3 \psi \ell_{\delta}$, with $\delta >0$. We conclude that $F$ behaves like $\psi$ \eqref{eq:psifin}, up to arbitrary small logarithmic correction.\footnote{We believe that the logarithmic correction is not really there and that $F/\psi$ is bounded.} This completes the proof of Theorem \ref{thm:main} for $t<t_1$ and $t\notin \mathcal T_\kappa$. 

Lastly, when $t$ belongs to the discrete set $\mathcal T_\kappa$ \eqref{eq:tk}, because of \eqref{eq:g0tk}, $g_0(u)$ \eqref{eq:g0} vanishes too fast at $u=0$, and neither an upper bound like \eqref{eq:psi0psi1} holds, nor Lemma \ref{lemma:sign}. But having established Theorem \ref{thm:main} for $t\notin \mathcal T_\kappa$ suffices to prove it for all $t$, by simply invoking the {\it convexity} property of the integral means spectrum \cite{Makarov}.

\section{Phase transitions}\label{phasetransitions}
\subsection{Loci of various spectra} In this section, we prove that  Figure \ref{fig} gives an accurate description of the phase transitions between different parts of the integral means spectrum. Since all $\beta$'s  and $\gamma$'s are given by simple algebraic equations, all inequalities are in principle elementary, but could be a bit fiddly if not addressed in the right way. %We assume in this section  that $t<0$. 

First, notice that $\gamma_0$ is  increasing and negative for $t<0$, whereas $\tg$ is decreasing and positive. Next, we shall need to study $\gamma_0+\tg$. Differentiating with respect to $t$ yields
$$
\frac{\partial}{\partial t}(\gamma_0+\tg)=
\frac{1}{\sqrt{(4+\kappa)^2/4-2t\kappa}}-\frac{1}{\sqrt{1-2t\kappa}}<0,
$$
so that $\gamma_0+\tg$ is decreasing; since $\gamma_0(0)=0$ and $\tg(0)=2/\kappa$, this gives that $\gamma_0+\tg>2/\kappa$ iff $t<0$.
\begin{proposition}
\label{pr:b}
Let $b=b(\gamma_0)$ be as in \eqref{eq:abc}, then $1/2-b>0$ if and only if $t>t_1$.
\end{proposition}
\begin{proof}
Owing to \eqref{eq:abc}, $b=\gamma_0+\tg-2/\kappa$. Solving the equation $b=1/2$ yields $t=t_1$. Since $b$ is decreasing with $t$, we obtain that $1/2-b>0$ if and only if $t>t_1$.
\end{proof}
\begin{proposition}\label{pr:tkn}
The set of equations, $b=n+\frac{1}{2}, n\in \mathbb N$, is realized at the finite set of points $\mathcal T_\kappa:=\{t_{1-n},  n\in \mathcal J_\kappa\}$, where $\mathcal J_\kappa:=\{n\in \mathbb N, 0\leq n\leq \lfloor\kappa^{-1}\rfloor\}$ and where $t_{1-n}$ is given by Eq. \eqref{eq:tkn}. 
\end{proposition}
\begin{proof} 
For $t\in (-\infty,0]$, we have that $b\in [0,\frac{1}{2}+\frac{1}{\kappa})$, which defines the possible range of values $n\in \mathcal J_\kappa$ where $b=n+\frac{1}{2}$. Solving the latter equation yields \eqref{eq:tkn}.
\end{proof}
\begin{proposition}
\label{pr:tb}
We have that
$
\tb(t)>\beta_0(t)
$
if and only if $t<t_1$. 
\end{proposition}
\begin{proof}
First of all, recall that all exponents $\beta$'s are given by the same quadratic function $\beta(\gamma)$ \eqref{eq:betagamma}. Since $\gamma_0\ne\tg$, $\beta(\tg)=\beta(\gamma_0)$ if and only if $\tg+\gamma_0=(4+\kappa)/2\kappa$. We know that this happens at $t=t_1$ only. Computing at $t=0$, we have $\beta_0(0)=0$ and $\tb(0)=-1$, hence $\tb(t)>\beta_0(t)$ if and only if $t<t_1$.
\end{proof}
Note that Propositions \ref{pr:b} and \ref{pr:tb} imply that the point where the construction in Ref. \cite{BeSmSLE} breaks down and  the point where $\tb$ exceeds $\beta_0$ coincide. 
\begin{proposition}
\label{pr:t_tip}
We have that 
$
\tb<\beta_{\mathrm{tip}}=\beta_0-2\gamma_0-1 
$ if and only if $t<0$.
\end{proposition}
%Note that, unlike two previous propositions, this is not ``if and only if'' statement.
\begin{proof}
 From Eqs. \eqref{eq:beta_not} and \eqref{eq:beta_tild} we have, 
\begin{comment}
$$
\begin{aligned}
\beta_0 & =-t+\frac{4+\kappa}{2}\gamma_0,\\
\tb & = -t-\frac{\kappa}{2}\tg,
\end{aligned}
$$
\end{comment}
$$
\beta_0  =-t+\frac{4+\kappa}{2}\gamma_0,\,\,\,
\tb  = -t-\frac{\kappa}{2}\tg,
$$
so that
$$
\beta_{\mathrm{tip}}-\tb=\frac{\kappa}{2}(\gamma_0+\tg)-1=\frac{\kappa}{2}b(\gamma_0),
$$
which is positive  iff $t<0$.
\end{proof}
\subsection{Second tip and derivative exponents}\label{Sectip}
  Of particular interest here is the \emph{packing spectrum}
%In the multifractal formalism, we have the usual transforms, with $p(s)$ the inverse function of $s(p)$:
 %=-\tau(s), f(\alpha)+\tau(s)=\alpha s, f'(\alpha)=s, \tau'(s)=\alpha.
%The moments of the harmonic measure $\mathcal H$ in a small ball of radius $\varepsilon$ scale as 
%$$\mathbb E [\mathcal H^s(\varepsilon)]\asymp \varepsilon^{\tau(s)}=\varepsilon^{-p(s)}. $$
associated with the $\beta_1$ spectrum of whole-plane SLE$_\kappa$ (see Ref. \cite{Makarov} 
for a detailed discussion of the different spectra of the
harmonic measure and their relations),
\begin{eqnarray}\label{eq:s1}s(t):=\beta_1(t)-t+1=-2t+\frac{1}{2}-\frac{1}{2}\sqrt{1-2\kappa t}.
\end{eqnarray}
In the domain $t\leq t_1<0$, $s(t)$ is decreasing;  its inverse is
 \begin{equation}\label{eq:nus} t=-\nu(s),\,\,\,\nu(s):=\frac{s}{2} +\frac{1}{16}\left(\kappa-4+\sqrt{(4-\kappa)^2+16 \kappa s}\right),
\end{equation}
where, as remarked in Ref. \cite[Section 4.4]{DNNZ},  $\nu(s)$ coincides with the non-standard multifractal tip exponents as obtained  in Ref. \cite[Eq. (12.19)]{MR2112128},  or with the so-called \emph{derivative exponents} as obtained in Ref. \cite[Eq. (3.1)]{MR2002m:60159b} for standard (interior or exterior) radial SLE. For $\hat g_\tau$ such a radial $\mathrm{SLE}_{\kappa}$ of hull $K_\tau$, the exponent $\nu$ describes the exponential decay $e^{-\nu\,\tau}$ in time $\tau>0$, of the moment of the boundary derivative modulus, $\mathbb E \left[|\hat g'_\tau(z)|^s\right]$, for $z \in \partial \mathbb D\setminus K_\tau$. It also governs the same exponential decay of the moment of order $s$, $\mathbb E \left[L_t^s\right]$, of the harmonic measure $L_t$ of $\partial \mathbb D\setminus K_\tau$ in $\mathbb D\setminus K_\tau$, as seen from the origin in the interior case, or from infinity in the exterior case. 

In the case of the interior whole-plane SLE, of map $f_0$ from $\mathbb D$ to a slit domain \cite{DNNZ}, one has  $f_0=\lim_{\tau \to +\infty} e^\tau \hat g_{-\tau}$, which is in law the same  as $\lim_{\tau \to +\infty} e^\tau \hat g_{\tau}^{-1}$, with $\hat g_\tau$ standard interior radial SLE;  in this limit, the unit circle is pushed back to infinity as $e^{\tau}\partial \mathbb D$.  Ref. \cite[Section 4.4, Figure 8]{DNNZ} then provides a heuristic explanation of the inverse relation between \eqref{eq:s1} and \eqref{eq:nus} as due, in the integral means $\int_{r\partial\mathbb D} |f_0'(z)|^t |dz|$ where $r\to 1^-$,  to the local boundary contribution of the image under $\hat g_{\tau}$ of  $\partial \mathbb D\setminus K_\tau$, i.e., in the limit $\tau\to +\infty$,  of the pre-image under $f_0$ of the point at infinity. 

%In the present case of the exterior whole-plane SLE, one has $g_0^{-1}=\lim_{\tau \to +\infty} e^{-\tau} \hat g_{-\tau}$, which is, in law, the same as $\lim_{\tau \to +\infty} e^{-\tau} \hat g_{\tau}^{-1}$, with $\hat g_\tau$ now standard exterior radial SLE;  the unit circle $\partial \mathbb D$  now shrinks to the origin as $e^{-\tau}\partial \mathbb D$. The inverse relation between \eqref{eq:s1} and \eqref{eq:nus}, can then be thought of as arising,  in the integral means $\int_{r\partial\mathbb D} |(g_0^{-1})'(z)|^t |dz|$ with $r>1$,  from the boundary contribution of the pre-image under $e^{-\tau} \hat g_{\tau}^{-1}$ of  $\partial \mathbb D\setminus K_\tau$, i.e., of the pre-image under $g_0^{-1}$ of (part of) a vanishingly small circle around the origin,  to which the SLE curve is anchored via a  `second tip'. 

The two whole-plane maps, interior $f_0$ and exterior $g_0^{-1}$, are naturally conjugate under the inversion map $z\mapsto 1/z$, as are the interior and exterior versions of radial SLE $\hat g_\tau$; in the exterior case, the unit circle shrinks as $e^{-\tau}\partial \mathbb D$ to  a vanishingly small circle around the origin, to which the whole-plane SLE curve is anchored via this  `second tip'.  This strongly suggests that for bounded whole-plane SLE, the $\beta_1$-spectrum is due to the presence of the second tip, image of the point at infinity in the unbounded case.

\section*{Acknowledgements}\thanks{The authors acknowledge partial funding in 2012 by Oxford Platform Grant BKRSZE0. They wish to thank the Simons Center for Geometry and Physics at Stony Brook University for its hospitality and support during the  Spring 2013 program ``Conformal Geometry''. They also wish to thank the Isaac Newton Institute (INI) for
Mathematical Sciences at Cambridge University for its hospitality and support during the 2015 program ``Random Geometry'', supported by EPSRC Grant Number EP/K032208/1. B.D. also gratefully acknowledges the support of a Simons Foundation fellowship at INI during the Random Geometry program. D.B. was partially funded by Engineering \& Physical Sciences
Research Council (EPSRC) Fellowship ref. EP/M002896/1.  B.D. acknowledges financial support from the French Agence Nationale
de la Recherche via the grant ANR-14-CE25-0014 ``GRAAL''; he is also partially funded by the CNRS Projet international de coop\'eration scientifique (PICS) ``Conformal Liouville Quantum Gravity'' n$^{\mathrm o}$PICS06769. B.D. and M.Z. are partially funded by the CNRS-{\sc insmi} Groupement de Recherche (GDR 3475) ``Analyse Multifractale''.}

\bibliography{sle}
\bibliographystyle{abbrv}

\end{document}